\documentclass[aps,prl,letterpaper,hyperlink,twocolumn]{revtex4}

\usepackage{amsmath, amsthm, amssymb,amstext,ulem,color}
\usepackage[colorlinks=true,citecolor=blue,linkcolor=magenta]{hyperref}

\theoremstyle{plain}

\newtheorem{lemma}{Lemma}

\newcommand{\ket} [1] {\vert #1 \rangle}

\newcommand{\norm}[1]{\left\lVert{#1}\right\rVert}

\def\bbE{\mathbb{E}}

\begin{document}

\title{A Note on ``Quantum Algorithm for Linear Systems of Equations'' }

\author{Yong-Zhen Xu$^1$, Yifan Huang$^1$, Zekun Ye$^1$ and Lvzhou Li$^{1,2,}$ \footnote{Electronic mail: lilvzh@mail.sysu.edu.cn (L.Li)} \vspace{0.2cm}}

\affiliation{$^1$ Institute of Computer Science Theory, School of Data and Computer Science, Sun Yat-sen University, Guangzhou 510006, China }
\affiliation{$^2$ The Key Laboratory of Machine Intelligence and Advanced Computing (Sun Yat-sen University) Ministry of Education, Guangzhou 510006, China}

\date{\today}

\begin{abstract}
Recently, an efficient quantum algorithm for linear systems of equations introduced by Harrow, Hassidim, and Lloyd, has received great concern from the academic community. However, the error and complexity analysis for this algorithm seems so complicated that it may
not be applicable to other filter functions for other tasks. In this note, a concise proof is proposed. We hope that it may inspire some novel HHL-based algorithms that can compute $F(A)\ket{b}$ for any computable $F$.
\end{abstract}

\maketitle

Solving linear systems of equations has been a central problem  in virtually all field of science and engineering. Recently, an efficient quantum algorithm for the problem was proposed by Harrow, Hassidim, and Lloyd ~\cite{HHL2009} (called HHL algorithm for short), which shows an exponential speed-up over the best known classical algorithm under certain conditions.  This algorithm has been considered to a new template showing how quantum computers
might be used to exponentially speed up certain problems, and may bring a series of applications, especially in the field of machine learning and big data \cite{Childs2009,Aaronson2015,BWPRWL2017}. Actually,  based on this seminal work~\cite{HHL2009}, some novel quantum algorithms were proposed, including Least-squares fitting~\cite{WBL2012}, Quantum support vector machine~\cite{RML2014}, Quantum PCA~\cite{LMP2014}, solving linear differential equations~\cite{Berry2014}, and so on.  It seems that how to find more nontrivial applications and further generalizations of the  work~\cite{HHL2009} has attracted much attention from  the academic community. In addition,
some other quantum algorithms using different ideas  have also been  presented for the linear systems problem~\cite{CJS2013,Ambainis2012,CKS2017,KP2017,WZP2017}.

Note that a full version of the paper~\cite{HHL2009} is Ref. \cite{0811.3171v3}. For consistency, we use the same symbols from Ref. \cite{0811.3171v3}. It is readily seen that the second inequality (A5) of Theorem 1 in Ref. \cite{0811.3171v3} is a core result for the error and complexity analysis of the HHL algorithm.
In the process of proving this result, Lemma 3 in Ref. \cite{0811.3171v3} plays a crucial role. However, both the proof for Lemma 3  and the proof for (A5) based on Lemma 3  seem too complicated, and they may not be applicable to other filter functions for other tasks. By the way, the proof of Lemma 3 was incomplete since $\tilde g > 0 $ was not considered when $\lambda \geq 1/\kappa$. A complete proof is given in the appendix which comfirms the correctness of Lemma 3. In this note, we propose a concise proof for (A5) based on Lemma 2 given by us. This new proof can alleviate the difficulties caused by the filter functions $f$ and $g$ in error analysis~\cite{0811.3171v3}. We hope that it may inspire some novel HHL-based algorithms which can compute $F(A)\ket{x}$ for any computable $F$. 

We start with two lemmas to be needed later.

\begin{lemma} \label{fg_Lip}
	The functions $f$ and $g$ are $O(\kappa)$-Lipschitz, meaning that for any $\lambda_1\neq\lambda_2$,
	\begin{align}\label{Ineq-a}
		\left| f\left( \lambda_1\right)  - f(\lambda_2)\right| \leq c_1\kappa \left| \lambda_1 - \lambda_2 \right|,
	\end{align}
	and
	\begin{align}\label{Ineq-b}
	\left|g(\lambda_1) - g(\lambda_2)\right|\leq c_1\kappa \left|\lambda_1 - \lambda_2\right|,
	\end{align}
	for  some $c_1= O(1)$.
\end{lemma}

\begin{proof}
	The two functions $f$ and $g$ are continuous and differentiable except at $\frac{1}{\kappa}$ and $\frac{1}{2\kappa}$, so we need only prove that the absolute value of the two derivatives are bounded. For all cases, the upper bounds of $\left| \frac{d}{dx}f(x) \right| $ and
	$ \left|\frac{d}{dx}g(x)\right| $  are $\frac{\pi\kappa}{2}$. Thus, $c_1 \geq \frac{\pi}{2}$. This completes the proof.
\end{proof}

\begin{lemma}\label{fg_sum}
	\begin{align}\label{Ineq-c}
	&\left|f(\lambda_1) - f(\lambda_2) \right|^2 + |g(\lambda_1) - g(\lambda_2)|^2 \le \nonumber\\
	 & c_2\kappa^2(\lambda_1-\lambda_2)^2 |f^2(\lambda_1) +
	 f^2(\lambda_2)+ g^2(\lambda_1)+g^2(\lambda_2)|,
	\end{align}
	for  some $c_2= O(1)$.
\end{lemma}

\begin{proof}
	 We need to consider nine cases since $f$ and $g$ are piecewise functions. However, since this inequality (\ref{Ineq-c}) has symmetry about $\lambda_1$ and $\lambda_2$, we only need to consider the following six cases when $\lambda_1 > \lambda_2$,
	 \begin{equation*}
	 \left\{
	 \begin{aligned}
	 &\text{case 1: } \lambda_1 \geq 1/\kappa, \lambda_2 \geq 1/\kappa, \\
	 &\text{case 2: } \lambda_1 \geq 1/\kappa, 1/2 \kappa \leq \lambda_2 < 1/\kappa, \\
	 &\text{case 3: } \lambda_1 \geq 1/\kappa, \lambda_2 < 1/ 2 \kappa, \\
	 &\text{case 4: } 1/2 \kappa \leq \lambda_1 < 1/\kappa, 1/2 \kappa \leq \lambda_2 < 1/\kappa, \\
	 &\text{case 5: }1/2 \kappa \leq \lambda_1 < 1/\kappa, \lambda_2 < 1/ 2 \kappa, \\
	 &\text{case 6: } \lambda_1 < 1/ 2 \kappa, \lambda_2 < 1/ 2 \kappa.
	 \end{aligned}
	 \right.
	 \end{equation*}
	
	 \textbf{Case 1:} we have
	 \begin{align}
	 &|f(\lambda_1) - f(\lambda_2)|^2 + |g(\lambda_1) - g(\lambda_2)|^2 \nonumber\\
	 &= \frac{1}{4k^2}\left(\frac{1}{\lambda_1}-\frac{1}{\lambda_2} \right) \nonumber\\
	 &= \frac{1}{4k^2}\frac{\left(\lambda_1 -\lambda_2\right)^2}{\lambda_1^2\lambda_2^2} \nonumber\\
	 &\leq \frac{1}{4}\frac{(\lambda_1 -\lambda_2)^2}{\lambda_1\lambda_2} \label{Ineq-d1}\\
	 &\leq \frac{1}{8}(\lambda_1 -\lambda_2)^2\left(\frac{1}{\lambda_1^2}+\frac{1}{\lambda_2^2}\right) \label{Ineq-d2} \\
	 &= \frac{k^2}{2}(\lambda_1 -\lambda_2)^2
	 \left|  f^2(\lambda_1)+f^2(\lambda_2)+g^2(\lambda_1)+g^2(\lambda_2) \right| \nonumber,
	 \end{align}
	 where Ineq. (\ref{Ineq-d1}) follows from $\frac{1}{\kappa^2} \leq \lambda_1\lambda_2$ and Ineq. (\ref{Ineq-d2}) follows from $\frac{1}{\lambda_1\lambda_2} \leq \frac{1}{2}\left(\frac{1}{\lambda_1^2}+\frac{1}{\lambda_2^2}\right)$.
	
	 \textbf{Cases 2-6:} In these cases, by Lemma~\ref{fg_Lip}, we have
	 \begin{align}
	& \left| f(\lambda_1) - f(\lambda_2)\right| ^2 + \left| g(\lambda_1) - g(\lambda_2)\right| ^2 \nonumber \\
	  &\leq 2c_1^2\kappa^2(\lambda_1-\lambda_2)^2.
	 \end{align}
	
	 In addition, the lower bound of $$ c_2\kappa^2\left( \lambda_1-\lambda_2\right) ^2 \left| f^2(\lambda_1) + f^2(\lambda_2)+ g^2(\lambda_1)+g^2(\lambda_2)\right|  $$ in these cases is
	 $\frac{1}{4}c_2\kappa^2(\lambda_1-\lambda_2)^2$.
     To prove the inequality (\ref{Ineq-c}), we only need $c_2 \geq 8c_1^2$.
	
	 In summary,  $c_2 \geq 8c_1^2$. This completes the proof.	
\end{proof}
Now we give the proof  for (A5) based on Lemma~\ref{fg_sum}.
\begin{proof}
	Recall that $\tilde{\lambda}_k := 2\pi k/t_0$, and $\delta_{jk} = t_0(\lambda_j - \tilde{\lambda}_k)$. We also abbreviate $f:=f(\lambda_j)$, $\tilde f := f(\tilde{\lambda}_k)$, $g:=g(\lambda_j)$ and $\tilde g =  g(\tilde{\lambda}_k)$. We define $p := \bbE [f^2 + g^2]$ and $\tilde{p} := \bbE [\tilde f^2 + \tilde g^2]$.
	
	In order to obtain an upper bound for $ \norm{\ket{x}  - \ket{\tilde x}}$, it suffices to give a lower bound  for $\langle x|\tilde x\rangle$, since it holds that $\norm{\ket{x}  - \ket{\tilde x}}= \sqrt{2(1-\textrm{Re}\langle x|\tilde x\rangle)}$. First, we have	
	\begin{align}\label{innerProduct}
	\langle x|\tilde x\rangle
	&= \frac{\bbE [f\tilde f + g \tilde g]}{\sqrt{p\tilde{p}}} \geq \frac{2\bbE [f\tilde f + g \tilde g]}{p + \tilde{p}},
	\end{align}
   where the inequality follows from
	$\sqrt{p\tilde{p}} \leq \frac{p+\tilde{p}}{2}$. Note that the inequality used here is different from one  in \cite{0811.3171v3}, which together with Lemma \ref{fg_sum} actually simplifies the proof of  (A5).

	Now we have
	\begin{align}
	& ( p + \tilde{p} ) \nonumber -2\bbE \left[f\tilde f + g \tilde g\right]   \nonumber \\
	=&\bbE\left[ |f - \tilde f|^2 + |g - \tilde g|^2\right] \label{eq1} \\
	\leq& \frac{c_2\kappa^2}{t_0^2}\bbE \left[ \delta_{jk}^2 \left( f^2 + \tilde{f}^2+g^2 + \tilde{g}^2 \right) \right] \label{eq2} \\
	\leq& \frac{c_2c_3\kappa^2}{t_0^2}\bbE \left[  \left( f^2 + \tilde{f}^2+g^2 + \tilde{g}^2 \right) \right] \label{eq3}\\
	=& \frac{c_2c_3\kappa^2}{t_0^2}(p+\tilde{p}),
	\end{align}
where Eq. (\ref{eq1}) follows from direct calculation, Ineq. (\ref {eq2}) follows from Lemma~\ref{fg_sum}, and Ineq. (\ref {eq3}) holds because the fact that each $\delta_{jk}^2$ is upper bounded by $c_3$  with $c_3=O(1)$. Therefore, we have
{\setlength\abovedisplayskip{1pt}
\setlength\belowdisplayskip{1pt}
\begin{align}\label{Effgg}
2\bbE \left[ f\tilde f + g \tilde g\right]
\geq \left( 1- \frac{c_2c_3\kappa^2}{t_0^2}\right)(p+\tilde{p}).
\end{align}}
Substituting (\ref{Effgg}) into (\ref{innerProduct}), we get
$\textrm{Re}\langle x|\tilde x\rangle\geq 1 -
O(\kappa^2/t_0^2)$. Hence, $\| \ket{\tilde x} -\ket{x}\| \leq \epsilon$. This completes the proof.
\end{proof}

\section{Appendix}

In this appendix, a detail proof for the Lemma 3 of Ref. \cite{0811.3171v3} is given. 
The Jordan's inequality to be used reads that $\frac{2}{\pi} \leq \frac{\sin x}{x} <1$
for $0 <|x|<\frac{\pi}{2}$.
\begin{proof} We prove the Lemma 3 by considering nine cases as follows.	
	\textbf{Case 1:} $\lambda \geq 1/\kappa$ and $\tilde{\lambda} \geq 1/\kappa$.
	\begin{align}
	&|f- \tilde f |^2 + |g- \tilde g |^2  \nonumber\\
	&= \frac{1}{4\kappa^2 \lambda^2 \tilde{\lambda}^2} \left( \frac{\delta}{t_0}\right) ^2\nonumber\\
	&\le \frac{1}{4\kappa^2 \lambda^2 } \kappa^2\left( \frac{\delta}{t_0}\right) ^2 \label{case1-1} \\
	&= \frac{\kappa^2}{t_0^2}\delta^2\left| f^2 + g^2\right|, 
	\end{align}
	where Ineq. (\ref {case1-1}) follows from $\frac{1}{\tilde{\lambda}^2} \le \kappa^2 $.

	\textbf{Case 2:} $\lambda \geq 1/\kappa$ and $1/2 \kappa \leq \tilde{\lambda} < 1/\kappa$.
	Let $\beta=\frac{\pi}{2}(2\kappa \tilde{\lambda} -1)$ and $\theta=\kappa \pi(\frac{1}{k}- \tilde{\lambda})$, that is $\beta = \frac{\pi}{2}-\theta$. Thus,
	$\sin \beta = \cos \theta$ and $\cos \beta = \sin \theta$. We have
	\begin{align}
	&	|f- \tilde f |^2 + |g- \tilde g |^2   \nonumber\\
	&	= \frac{1}{4} \left(  \frac{1}{\kappa ^2 \lambda ^2} - \frac{2}{\kappa \lambda} \sin \beta+1 \right)  \nonumber\\
	&= \frac{1}{4}\left( \left( \frac{1}{\kappa \lambda} \sin \beta-1\right) ^2+\frac{1}{\kappa^2\lambda^2}\cos ^2 \beta\right)  \nonumber\\
	&= \frac{1}{4\kappa^2 \lambda^2} \left( (\kappa \lambda- \cos \theta) ^2 +\sin^2 \theta \right)  \nonumber\\
	&= \frac{1}{4\kappa^2 \lambda^2} \left( (\kappa \lambda -1  +  1- \cos \theta)^2 + \sin^2 \theta\right)  \nonumber\\
	&\leq \frac{1}{4\kappa^2 \lambda^2} \left( 2(\kappa \lambda - 1)^2+2(1- \cos \theta)^2 +\sin^2 \theta \right) \nonumber\\
	& \leq  \frac{1}{4\kappa^2 \lambda^2} \left( 2(\kappa \lambda - 1)^2+8(\sin^2 \frac{\theta}{2})^2  +\sin^2 \theta \right) \nonumber\\
	&\leq\frac{1}{4\kappa^2 \lambda^2} \left( 2(\kappa \lambda - 1)^2+\frac{\theta ^4}{2}+ \theta ^2\right)  \nonumber\\
	&= \frac{1}{4\kappa^2 \lambda^2} \left( 2(\kappa \lambda - 1)^2+\frac{\pi ^4}{2} (1-\kappa  \tilde{\lambda})^4   + \pi ^2(1-\kappa \tilde{\lambda})^2\right)  \nonumber\\
	&\leq \frac{1}{4\kappa^2 \lambda^2} \left( 2(\kappa \lambda - 1)^2+\frac{\pi ^4}{2}(1-\kappa \tilde{\lambda})^2   + \pi ^2(1-\kappa \tilde{\lambda})^2\right)  \nonumber\\
	&\leq \frac{1}{4\kappa^2 \lambda^2} \left( \frac{\pi ^4}{2}+ \pi ^2\right)  \left( (\kappa \lambda - 1)^2+  (1-\kappa \tilde{\lambda})^2\right)  \nonumber\\
	&\leq \frac{1}{4\kappa^2 \lambda^2} \left( \frac{\pi ^4}{2}+ \pi ^2\right)  \left( (\kappa \lambda - \kappa \tilde{\lambda})^2\right)  \nonumber\\
	&= \left( \frac{\pi ^4}{2}+ \pi ^2\right) \frac{\kappa^2}{t_0^2}\delta^2\left| f^2 + g^2\right|, 
	\end{align}
	where the second inequality follows from the half-angle formula for cosine functions ,
	the third inequality follows from Jordan's inequality and others follow from direct calculation.
	
	\textbf{Case 3:} $\lambda \geq 1/\kappa$ and $\tilde{\lambda} < 1/ 2 \kappa$. We have
	\begin{align}
	\frac{\kappa^2 \lambda^2+1}{\kappa^2 (\lambda -\tilde{\lambda})^2}
	\leq \frac{2\kappa^2 \lambda^2}{\kappa^2 (\lambda -\tilde{\lambda})^2} 
	=2(\frac{1}{1-\frac{\tilde{\lambda}}{\lambda}})^2  
	\leq 8,
	\end{align}
	where the first inequality follows from $1 \leq \kappa^2 \lambda^2$ and the second inequality follows 
	from $\frac{\tilde{\lambda}}{\lambda} < \frac{1}{2}$. Thus, we have
    $\frac{1}{4}(1+\frac{1}{\kappa^2\lambda^2}) \leq 8 \frac{\kappa^2 (\lambda -\tilde{\lambda})^2}{4\kappa^2 \lambda^2}$.
	That is,
	\begin{align}
	|f- \tilde f |^2 + |g- \tilde g |^2 	
	&< 8\frac{\kappa^2}{t_0^2}\delta^2\left| f^2 + g^2\right|. 
	\end{align}
	
	\textbf{Case 4:} $1/2 \kappa \leq \lambda < 1/\kappa$ and $\tilde{\lambda} \geq 1/\kappa$.
	Similar to case 2, we have the same result.
	
	\textbf{Case 5:} $1/2 \kappa \leq \lambda < 1/\kappa$ and $1/2 \kappa \leq \tilde{\lambda} < 1/\kappa$.
	\begin{align} 	
	|f- \tilde f |^2 + |g- \tilde g |^2 
	=& \frac{1}{2} \left( 1 - \cos\left( \pi \kappa(\lambda - \tilde{\lambda})\right)  \right)  \nonumber \\
	=& \sin ^2\left( \frac{\pi \kappa(\lambda - \tilde{\lambda})}{2}\right) \label{case5-1} \\
	<& \frac{\pi^2}{4}\kappa^2\left( \lambda - \tilde{\lambda}\right) ^2 \label{case5-2} \\
	=& \pi^2 \frac{\kappa^2}{t_0^2}\delta^2\left| f^2 + g^2\right|, 
	\end{align}
	where Eq. (\ref{case5-1}) is based on the half-angle formula for cosine functions and 
	Ineq. (\ref {case5-2}) follows from Jordan's inequality.
	
	\textbf{Case 6:} $1/2 \kappa \leq \lambda < 1/\kappa$ and $\tilde{\lambda} < 1/ 2 \kappa$.
	\begin{align}
	|f- \tilde f |^2 + |g- \tilde g |^2 
	& =\frac{1}{2}\left( 1 - \cos \left( \frac{\pi}{2}(2\kappa \lambda -1)\right)  \right) \nonumber\\
	&= \sin ^2 \left( \frac{\pi}{4}(2\kappa \lambda -1)\right) \label{case6-1}\\
	&\leq \left( \frac{\pi}{4}(2\kappa \lambda -1)\right) ^2 \label{case6-2}\\
	&= \left( \frac{\pi}{2}(\kappa \lambda -1/2)\right) ^2 \nonumber\\
	&\leq \left( \frac{\pi}{2}(\kappa \lambda -\kappa \tilde{\lambda})\right) ^2 \label{case6-3}\\
	&= \frac{\pi^2}{4}\kappa^2\left( \lambda - \tilde{\lambda}\right) ^2 \nonumber \\
	&= \pi^2 \frac{\kappa^2}{t_0^2}\delta^2\left| f^2 + g^2\right|, 
	\end{align}
	where Eq. (\ref {case6-1}) holds on account of the half-angle formula for 
	cosine functions, Ineq. (\ref {case6-2}) follows from Jordan's inequality 
	and Eq. (\ref {case6-3}) follows from $\kappa \tilde{\lambda} < 1/2$.
	
	\textbf{Case 7:} $\lambda < 1/ 2 \kappa$ and $\tilde{\lambda} \geq 1/\kappa$.
	Similar to case 3, we have the same result.
	
	\textbf{Case 8:} $\lambda < 1/ 2 \kappa$ and $1/2 \kappa \leq \tilde{\lambda} < 1/\kappa$.
	Similar to case 6, we have the same result.
	
	\textbf{Case 9:} $\lambda < 1/ 2 \kappa$ and $\tilde{\lambda} < 1/ 2 \kappa$. The Lemma 3 holds for $c\geq 0$.
	
	In summary, $c \geq \frac{\pi^4}{2} + \pi^2$. This completes the proof.
\end{proof}
\end{document}